\documentclass[english,11pt]{article}

\usepackage[latin9]{inputenc}
\usepackage{color}
\usepackage{amsthm}
\usepackage{amsmath}
\usepackage{amssymb}
\usepackage{fullpage}

\usepackage[colorlinks,
            linkcolor = blue,
            citecolor = blue,
            urlcolor = blue]{hyperref}

\makeatletter

\theoremstyle{plain}
\newtheorem{thm}{\protect\theoremname}

\newtheorem{lem}[thm]{\protect\lemmaname}

\theoremstyle{definition}
\newtheorem{defn}[thm]{\protect\definitionname}

\def\01{\{0,1\}}

\newcommand{\eps}{\varepsilon}
\newcommand{\ket}[1]{|#1\rangle}
\newcommand{\bra}[1]{\langle#1|}
\newcommand{\ketbra}[2]{|#1\rangle\langle#2|}

\newcommand{\inpc}[2]{\langle{#1},{#2}\rangle} % inproduct, < , >
\newcommand{\set}[1]{\{#1\}}
\newcommand{\abs}[1]{|#1|}
\newcommand{\Tr}{\mbox{\rm Tr}}

\newcommand{\C}{\mathbb{C}}

\newcommand{\N}{{\cal N}}

\makeatother

\usepackage{babel}
  \providecommand{\definitionname}{Definition}
  \providecommand{\lemmaname}{Lemma}
  \providecommand{\propositionname}{Proposition}
  \providecommand{\theoremname}{Theorem}
  \providecommand{\corollaryname}{Corollary}

\begin{document}

\title{A Generalization of Kochen-Specker Sets Relates Quantum Coloring to Entanglement-Assisted Channel Capacity}

\author{
Laura Man\v{c}inska
\thanks{IQC, University of Waterloo, 200 University Ave West, Waterloo ON, N2L 3G1, Canada. Email: lmancins@uwaterloo.ca }
\and
Giannicola Scarpa
\thanks{CWI, Science Park 123, 1098 XG Amsterdam, the Netherlands. Supported by a Vidi grant from the Netherlands Organization for Scientific Research (NWO). \mbox{Email: g.scarpa@cwi.nl}}
\and
Simone Severini
\thanks{Department of Computer Science, and Department of Physics \& Astronomy, University College London, WC1E 6BT London, United Kingdom. Supported by the Royal Society. Email: simoseve@gmail.com}
}

\maketitle

\begin{abstract}
We introduce two generalizations of Kochen-Specker (KS) sets: projective KS sets and generalized KS sets. We then use projective KS sets to characterize all graphs for which the chromatic number is strictly larger than the quantum chromatic number. Here, the quantum chromatic number is defined via a nonlocal game based on graph coloring. We further show that from any graph with separation between these two quantities, one can construct a classical channel for which entanglement assistance increases the one-shot zero-error capacity. As an example, we exhibit a new family of classical channels with an exponential increase.
\end{abstract}

\section{Introduction}

We establish a connection between three topics that involve the
foundations of quantum mechanics and quantum Shannon theory. These are
generalized Kochen-Specker sets, entanglement-assisted zero-error
capacities, and graph coloring games.

The (Bell-)Kochen-Specker theorem \cite{Bell-KS, KS} is one of the most striking results in quantum mechanics.
Informally, it states that it is not possible for all quantum mechanical observables to have definite values that are also noncontextual (i.e., independent of the measurement arrangement).
One way to prove the theorem is by exhibiting a set of
vectors for which there is no function that singles out exactly one vector for each orthonormal basis in the set.
Such collections are called \emph{Kochen-Specker} (KS) \emph{sets}
(see, \emph{e.g.}, \cite{peres,cabKS18}).

Recently a relaxation of KS sets has been used to construct classical channels with quantum advantage \cite{CLMW09}. Consider two parties that wish to use a noisy classical channel to transfer information with zero probability of error. For some channels quantum players using shared entanglement can outperform their classical counterparts. More precisely, it has been shown that entanglement can increase both the one-shot \cite{CLMW09} and asymptotic \cite{LMMOR, BBG11} zero-error capacity. The entanglement-assisted zero-error capacity is the quantum analogue of Shannon zero-error capacity, a well-studied concept in classical information theory. The entanglement-assisted zero-error capacity is known to be upper bounded by the Lov\'{a}sz theta number \cite{ZEC-theta, Winter}, which can be computed efficiently.

A relaxation of KS sets has also been shown to characterize a certain class of \emph{pseudo-telepathy games }\cite{PT-survey, KS-PT}. These are nonlocal games that can be won with certainty by quantum players using shared entanglement, while all classical players  fail with nonzero probability.
A specific class of pseudo-telepathy games is based on graph coloring.
Here, the advantage of entangled players is quantified by the so-called \emph{quantum chromatic number}. The quantum chromatic number can also be interpreted combinatorially via a vertex labeling problem using tuples of projectors satisfying certain completeness and orthogonality constraints \cite{qchrom}. Remarkably, the computational complexity of the quantum chromatic number is unknown.

The present work relates the above notions.
It gives a description of different phenomena with the same mathematical toolbox,
with
potential applications in zero-error information theory such as
nonclassical communication protocols \cite{superduper}, quantum analogues of information-hiding methods
(\emph{e.g.}, watermarking \cite{watermark}), and novel proof techniques for combinatorial
problems either in optimization or in extremal set theory \cite{Alon}.

\subsection{Our contribution}

We introduce two generalizations of Kochen-Specker sets:
\emph{projective KS sets}, consisting of projectors, and \emph{generalized KS sets}, consisting of positive-semidefinite operators.
We show that projective KS sets are useful in a number of ways:
\begin{enumerate}
\item
In Section \ref{sec:zec2}, we show how to use projective KS sets to construct classical channels for which entaglement assistance increases their one-shot zero-error capacity. This generalizes a similar construction for KS sets given in \cite{CLMW09}.
\item In Section \ref{sec:qcn2}, we show that projective KS sets completely characterize the graphs exhibiting a separation between chromatic number and quantum chromatic number. The characterization settles the grap-theoretic discussion started in \cite{qchrom}; a characterization for the rank-1 case was already proposed in \cite{SS12}. Interestingly, \cite{qchrom-aqis} observed a separation between rank-1 and general rank quantum chromatic number. Hence, the use of projective KS sets is indeed necessary for a full characterization.
\item In Section \ref{sec:QCN_ZEC}, we show how to use graphs for which the chromatic and quantum chromatic number are different to construct classical channels with a separation between entanglement-assisted and classical one-shot zero-error capacity.
This answers an open question formulated in \cite{SS12}. Graphs with such separations are useful for quantifying the information-theoretic gain permitted by the use of shared entanglement.
\item In Appendix \ref{apx:pt}, we show that there is a correspondence between projective KS sets and pseudo-telepathy games that can be won by quantum players using only projective measurements on maximally entangled state. This generalizes a similar result about KS sets given in \cite{KS-PT}.
\end{enumerate}

It is still an open problem to find direct applications of generalized KS sets in the context of quantum Shannon theory.
Through the paper assume familiarity with the basics of quantum
information theory. The reader can find a good introduction in \cite[Chapter 2]{NielsenChuang}.

\section{Kochen-Specker sets}

Let $S \subset \C^n $. A function \mbox{$f: S \rightarrow \{0,1\}$} is a \emph{marking function} for $S$ if
for all orthonormal bases $b \subset S$ we have $\sum_{u\in b} f(u) = 1$.
Gleason's theorem \cite{Gleason} implies that for any $n\geq 3$ there does not exist a marking function for $\C^n$. Bell \cite{Bell-KS}, and independently Kochen and Specker \cite{KS}, interpreted this statement in the framework of contextuality of physical theories. For this reason, this statement is also known as the (Bell-)Kochen-Specker theorem. Since then, finite sets of vectors in some given dimension giving rise to a proof of this theorem are known as Kochen-Specker (KS) sets.

\begin{defn} [KS set]
A set of unit vectors $S\subset \C^n$ is a \emph{Kochen-Specker set} if there is no marking function for $S$.
\end{defn}

Renner and Wolf \cite{KS-PT} considered a generalization of KS sets called \emph{weak} KS sets. Intuitively, for a weak KS set there can be marking functions, but every such function evaluates to $1$ for two orthogonal vectors in the set.
\begin{defn} [weak KS set]
A set of unit vectors $S\subset \C^n$ is a \emph{weak Kochen-Specker set} if for all marking functions $f$ for $S$ there exist orthogonal vectors $u,v\in S$ such that $f(u)=f(v)=1$.
\end{defn}
As explained in \cite{KS-PT}, every KS set is clearly a weak KS set but the converse does not always hold. However, every weak KS set can be completed to a KS set by adding $O(|S|^2)$ elements. Hence, a weak KS set also gives a proof of the Kochen-Specker theorem in some specific dimension. In fact it is more convenient to deal with weak KS sets, since they capture the essence of KS sets and can contain fewer vectors.

\subsection{Generalizations of KS sets}
We now define a natural generalization of weak KS sets by considering
subsets of $\mathcal{Q}_n$, the set of all $n\times n$ orthogonal projectors, instead of
subsets of $\mathbb{C}^{n}$.
Recall that an orthogonal projector is a Hermitian matrix $P$ such that $P^2 = P$.
For brevity,
from now on we omit the word ``orthogonal'' when we talk about projectors. Moreover,
in naming this generalization of KS sets we omit the word ``weak'' and use the term ``projective KS sets''.

Let $\mathcal{M}_n$ be the set of $n\times n$ matrices.
A \emph{marking function
$f$ for $S\subset\mathcal{M}_n$ }is a function $f:S\rightarrow\{0,1\}$ such
that for all $M \subset S$ with $\sum_{P\in M}P=I$,
we have $\sum_{P\in M}f(P)=1$.

\begin{defn}  [Projective KS set]
A set $S\subset\mathcal{Q}_n$ is a \emph{projective
Kochen-Specker set}  if for all \emph{marking functions} $f$ for $S$,
there exist $P,P'\in S$ for which $\Tr(PP')=0$
and $f(P)=f(P')=1$.
\end{defn}

It is easy to see that each set $M \subset \mathcal{Q}_n$ with $\sum_{P\in M}P=I$ is a projective measurement.
Also, note that weak KS sets are a special case of projective KS sets, if we identify a vector with the corresponding rank-$1$ projector.
Conversely, starting from any projective KS set one can construct (usually infinitely many) underlying weak KS sets (see Appendix \ref{apx:KS}).

Although in the rest of the paper we will only deal with projective KS sets, we can further generalize weak KS sets by considering subsets of  $\mathcal{S}_{+}^{n}$, the set of all $n\times n$ positive semidefinite matrices.

\begin{defn} [Generalized KS set]
\label{def:genKS}
A set $S\subset\mathcal{S}_{+}^{n}$ is a \emph{generalized Kochen-Specker set} if for all \emph{marking functions} $f$ for $S$
there exist $E,E'\in S$ for which $E+E'\leq I$
and $f(E)=f(E')=1$.
\end{defn}

Note that each set $M \subset \mathcal{S}_{+}^{n}$ with $\sum_{E\in M}E=I$ is a POVM (where, as usual, ``POVM'' stands for positive operator-valued measure). Projective KS sets are a special case of generalized KS sets, because when $S$ is a set of projectors the condition $E+E'\leq I$ is equivalent to $\Tr(EE')=0$.

KS-like sets consisting of positive semidefinite matrices have already been considered by Cabello \cite{Cabello}. Motivated by a a recent analogue of Gleason's theorem for positive semidefinite operators in two dimensions \cite{Busch,Fuchs}, Cabello exhibits what we here call a generalized KS set in $\mathcal{S}^2_+$. Hence, generalized KS sets exist even in two dimensions and have turned out to be useful for scenarios where regular KS sets do not apply (recall that there are no KS sets in $\C^2$).

\section{Zero-error channel capacity} \label{sec:zec}

In this section, we explain the concepts of classical and entanglement-assisted zero-error capacity
of a channel. In what follows, every graph is unweighted, undirected and without self-loops.

A \emph{classical channel} $\mathcal{N}$ with input set $X$ and output
set $Y$ is specified by a conditional probability distribution $\mathcal{N}(y|x)$, the probability to produce output $y$ upon input $x$.
Two inputs $x,x'\in X$ are \emph{confusable} if
there exists $y\in Y$ such that $\mathcal{N}(y|x)>0$ and $\mathcal{N}(y|x')>0.$
We then define the \emph{confusability graph of channel $\mathcal{N}$}, $G(\mathcal{N})$, as
the graph with vertex set $X$ and edge set $\{(x,x'):x,x'\mbox{ are distinct and confusable}\}$.

The \emph{one-shot zero-error capacity} of $\mathcal N$, $c_0(\mathcal{N})$,
is the size of a largest set of nonconfusable
inputs. This is just the independence number $\alpha(G(\mathcal{N}))$ of the confusability
graph.
In the \emph{entanglement-assisted}
setting, the sender (Alice) and receiver (Bob) share an entangled state $\rho$ and can perform local quantum measurements on their part of  $\rho$.

Let us now describe the general form of an entanglement-assisted protocol used by Alice to send one out of $q$ messages to Bob with a single use of the classical channel $\mathcal{N}$ (also see \cite{CLMW09}). For each message $m\in [q]$, Alice has a POVM $\mathcal{E}^m = \{E_{1}^{m},\dots,E_{|X|}^{m}\}$ with $|X|$ outputs.
To send message $m$, she measures her subsystem
using $\mathcal{E}^m$ and sends through the channel the
observed $x\in X$.
Bob receives some $y\in Y$ with $\mathcal{N}(y|x)>0$.
If the right condition holds (as we will explain below), Bob can recover $m$ with certainty
using a projective measurement on his subsystem.

We now state the necessary and sufficient condition for the success of the protocol.
If Alice gets outcome $x\in X$ upon measuring $\mathcal{E}^m$,
Bob's part of the entangled state collapses to
$\beta_{x}^{m}=\Tr_{A}((E_{x}^{m}\otimes I)\rho)$. Given channel's output $y$, Bob can recover $m$ if and only if
\begin{equation}
\forall m\neq m',\forall\mbox{ confusable }x,x'\ \Tr(\beta_{x}^{m}\beta_{x'}^{m'})=0.\label{eq:condition}
\end{equation}
Bob can recover the message with a projective measurement on the mutually orthogonal supports of
$$ \sum_{x \;:\; \mathcal{N}(y|x)>0} \beta_{x}^{m},$$
for all messages $m$.
In such a case, we say that assisted by the entangled state $\rho$ Alice can use the POVMs $\mathcal{E}^1,\dotsc,\mathcal{E}^q$ as her strategy for sending one out of $q$ messages with a single use of $\mathcal{N}$.

\begin{defn}
The \textit{entanglement-assisted one-shot zero-error channel capacity }of $\mathcal{N}$,  $c_0^*(\mathcal{N})$,  is the maximum integer $q$ such that
there exists a protocol for which condition \eqref{eq:condition} holds.
\end{defn}

Considering more than a single use of the channel, one can define the asymptotic zero-error channel capacity $\Theta(\mathcal{N})$ and the asymptotic entanglement-assisted zero-error channel capacity $\Theta^*(\mathcal{N})$ by
$ \displaystyle \Theta(\mathcal{N}) :=
        \lim_{k\rightarrow\infty} (c_0(\mathcal{N}^{\otimes k}))^{1/k}
 $ and $
\displaystyle
  \Theta^*(\mathcal{N}) :=
        \lim_{k\rightarrow\infty} (c_0^*(\mathcal{N}^{\otimes k}))^{1/k}.
$
Later in the paper we directly deal only with one-shot capacities.

Since $c_0^*(\mathcal{N})$ depends solely on the confusability graph $G(\mathcal{N})$ \cite{CLMW09}, we can talk about $c_0^*(G)$ for a graph $G$, meaning the entanglement-assisted one-shot zero-error capacity of a channel with confusability graph $G$. Similarly we can talk about quantities $c_0(G), \Theta(G),$
and $\Theta^*(G)$.

\subsection{Relationship with KS sets} \label{sec:zec2}

As shown in \cite{CLMW09}, every weak Kochen-Specker
set can be used to construct a graph for which $c_0(G)<c_0^*(G)$.
Following the line of argument of  \cite{CLMW09}, we now prove that also projective KS sets can be used for this purpose (see Theorem \ref{thm:KStoG}) as well as a weak converse of this statement (see Theorem \ref{thm:GtoKS}).

Let $\ket{\Psi} := \frac{1}{\sqrt n} \sum_{i\in [n]} \ket{i}\ket{i}$, where $\{\ket{i}\}_{i\in [n]}$ is the standard basis of $\C^n$. We start by proving the following technical lemma.
\begin{lem} \label{lem:ZEC-graph-partition}
Let $G=(V,E)$ be a graph whose vertex set $V$ can be partitioned into $k$ cliques $S_1, \dots, S_k$, not necessarily of the same size.
Assume that there is an assigment of a projector $P_v$ to each vertex $v$ such that:
\begin {enumerate}
 \item for all edges $(v, w)$, we have $ \Tr(P_v P_w) =0$,
 \item for all $i\in [k]$, we have $\sum_{v\in S_i} P_v = I$.
\end {enumerate}
Then, $c_0^*(G)\geq k$.
\end{lem}

\begin{proof}
We design a strategy for Alice and Bob to send one out of $k$ messages over a channel with
confusability graph $G$.
Let Alice and Bob share the maximally entangled state $\ket{\Psi}$ defined above.
To send message $m\in[k]$, Alice measures her part of $\ket{\Psi}$
using the projective measurement given by the projectors assigned to $S_{m}$ and sends the outcome $v\in V$
through the channel.
For all $v, m$ we have that either $\beta_{v}^{m} =\Tr_{A}((P_v\otimes I)\ketbra{\Psi}{\Psi}) = P_v^\top /n$ if $v\in S_m$ or $\beta_{v}^{m}=0$ if $v\notin S_m$. Since the projectors assigned to adjacent (confusable) vertices are orthogonal, the Condition \eqref{eq:condition} holds. Hence, the above strategy is valid and $c_0^*(G)\geq k$.
\end{proof}

We relate every multiset of projectors $T$ to a graph.
The \emph{orthogonality graph} of $T$ is the
graph with vertex set $T$ and edge set $\{P,P': \Tr(PP')=0\}$. Let us denote by $\uplus$ the multiset union. We are now ready to prove the following.
\begin{thm}[Projective KS set $\Rightarrow$ separation]
Let $S$ be a projective KS set.
Let $S_{1},\dots,S_{k} \subset \mathcal{Q}_n$ be all the subsets of $S$ such that $\sum_{P\in S_{i}}P=I$. Then the orthogonality graph $G$ of the multiset $S_{1} \uplus \dots \uplus S_{k}$ satisfies $c_0(G)<c_0^*(G)$.
\label{thm:KStoG}
\end{thm}

\begin{proof}
Observe that every $S_i$ is a projective measurement,
so the vertices of $G$ can be partitioned in $k$ cliques $S_{1},\dots,S_{k}$.
Let $T$ be a maximal independent set in $G$. Suppose towards a contradiction that $\abs{T}=k$, i.e., $T$ is a multiset of projectors containing exactly one element per clique. From maximality of $T$, we have that if $P\in T$ is part of $\ell$ measurements, then $T$ contains $\ell$ copies of $P$. Define a marking function for $S$ as:
$$
f(P) = 1 \iff P \in T.
$$
It is a marking function for $S$ because  $S_{1},\dots,S_{k} \subset \mathcal{Q}_n$ are all the projective measurements in $S$ and $f$ selects exactly one element from each $S_i$. Moreover, $f$ does not mark any pair of orthogonal elements because $T$ is an independent set and $G$ is an orthogonality graph.
The existence of $f$ contradicts the assumption that $S$ is a projective KS set. Therefore, $c_0(G)= \alpha(G) < k$.
To see that $c_0^*(G)\geq k$, partition the vertices of $G$ into $k$ cliques $S_{1},\dots,S_{k}$ and use Lemma \ref{lem:ZEC-graph-partition}.
\end{proof}

The following theorem provides a weak converse of Theorem \ref{thm:KStoG}. It relates a particular class of classical channels to projective KS sets, namely the channels for which Alice can send one out of $c_0^*(\mathcal{N})$ messages using projective measurements on $\ket{\Psi}$ as her strategy.
To the best of our knowledge, it is an open question whether projective measurments and maximally entangled state are always sufficient to achieve $c_0^*(\mathcal{N})$.

\begin{thm}[Separation $\Rightarrow$ projective KS set]
Let $G=(V,E)$ be the confusability graph of a classical channel $\mathcal{N}$ and $k\in \mathbb{Z}$ be such that $c_0(G)<k$. If Alice can use the projective measurements $\Pi^m=\set{P_{1}^{m},\dots,P_{|V|}^{m}}$ for all ${m\in[k]}$ together with $\ket{\Psi}$ as her strategy for sending one out of $k$ messages over $\mathcal{N}$ then
\[
  S:=\set{P_{v}^{m}: v\in [\abs{V}], m\in[k]}
\]
is a projective KS set.
\label{thm:GtoKS}
\end{thm}

\begin{proof}
We prove that if $S$ is\emph{ not} a projective KS set, then we can construct an independent set of size $k$ and thus $c_0(G)\geq k$. If $S$ is not a projective KS set, then there exists a marking function $f:S\rightarrow\{0,1\}$ such that for any $P,P'\in S$ for which $\Tr(PP')=0$, $f(P)=0$ or $f(P')=0$. Consider the set
\[
  J := \{v\in V : f(P_{v}^{m})=1, \text{ for some } m\in[k]\}.
\]
We now show that $J$ is an independent set of size $k$. Note that in a protocol where Alice measures her part of $\ket{\Psi}$ using $\Pi^m$, Bob's state is
\[
  \beta_{v}^{m} =\Tr_{A}((P_v^m\otimes I)\ketbra{\Psi}{\Psi}) 
= (P_v^m)^\top /n
\]
if Alice gets outcome $v$. Since Alice can use this strategy for communicating messages with zero-error, Condition~\ref{eq:condition} is respected and we get that for all distinct $m, m'\in[k]$:
\begin{enumerate}
\item $\Tr(\beta_v^m\beta_v^{m'})=0$ and hence $\Tr(P_v^{m} P_v^{m'})=0$;
\item for all $(u,v)\in E$ we have $\Tr(\beta_u^m\beta_v^{m'})=0$ and hence $\Tr(P_u^{m} P_v^{m'})=0$. 
\end{enumerate}

Recall that $f$ cannot assign value one to two orthogonal projectors. Hence, from (1.)~and the fact that $f$ selects one projector from each of $k$ measurements $\Pi^m$, we obtain that $\abs{J}=k$. Similarly, from (2.)~and the fact that $\Tr(P^m_v P^m_u)=0$ for all distinct $u,v\in V$, we get that $J$ is an independent set.
\end{proof}

We remark that Theorems \ref{thm:KStoG} and \ref{thm:GtoKS} also follow from Theorem \ref{thm:ks_pt} in the Appendix, by using a reduction from one-shot zero-error channel capacity to pseudo-telepathy games given in \cite{CLMW09}. However, the direct approach taken in the proofs above more clearly shows the relationship between orthogonality graphs of projective KS sets and graphs having $c_0(G)<c_0^*(G)$.

\section{Chromatic number} \label{sec:qcn}

In this section we illustrate the concept of quantum chromatic number of a graph \cite{qchrom} and state some of its properties. Then, we prove that projective Kochen-Specker sets completely characterize the class of graphs for which this parameter is strictly larger than the chromatic number.

A \emph{proper $c$-coloring} of a graph is an assignment of $c$
colors to the vertices of the graph such that every two adjacent vertices
have different colors. The \emph{chromatic number} of a graph $G$,
denoted by $\chi(G)$, is the minimum number of colors $c$ such that
there exists a proper $c$-coloring of $G$.

Define the \emph{coloring game} for $G=(V,E)$ as follows. Two players,
Alice and Bob, claim that they have a proper $c$-coloring for $G$.
A referee wants to test this claim with a one-round game, so he forbids
communication between the players and separately asks Alice the color
$\alpha$ for the vertex $v$ and Bob the color $\beta$ for the vertex
$w$. The players win the game if the following holds:
\begin{enumerate}
\item If $v=w$, then $\alpha=\beta$
\item If $(v,w)\in E$, then $\alpha\neq\beta$
\item $\alpha,\beta\in\{1,\dots,c\}$.
\end{enumerate}
A classical strategy consists of two deterministic functions $c_{A}:V\rightarrow[c]$
for Alice and \mbox{$c_{B}:V\rightarrow[c]$} for Bob. To win with probability one, we must have $c_{A}=c_{B}$
(to satisfy the first condition) and $c_{A}$ must be a valid
$c$-coloring of the graph (to satisfy the second and third conditions).
Therefore, classical players cannot win the game with probability one using fewer than $\chi(G)$ colors.

A quantum strategy for the coloring game uses an entangled state $\ket{\psi}$ and two families of POVMs: for all $v\in V$,
Alice has a POVM $\mathcal{E}^v = \{E^v_{\alpha}\}_{\alpha=1,\dots,c}$ and Bob has a POVM $\mathcal{F}^v =\{F^v_{\beta}\}_{\beta=1,\dots,c}$.
Upon input $v$, Alice measures her part of $\ket{\psi}$ with $\mathcal{E}^m$ and responds with the obtained outcome $\alpha$. Bob acts similarly and outputs
$\beta$. The requirements for the game translate into the following
\emph{consistency conditions}. The players win the coloring game
with certainty if and only if
\begin{align}
&\forall v\in V,\forall\alpha\neq\beta,\ \bra{\psi}E^v_{\alpha}\otimes F^v_{\beta}\ket{\psi}=0\label{eq:consistency1} \\
&\forall(v,w)\in E,\forall\alpha,\ \bra{\psi}E^v_{\alpha}\otimes F^w_{\alpha}\ket{\psi}=0.\label{eq:consistency2}
\end{align}
 In this case, we call the strategy a \emph{winning strategy} or a
\emph{quantum $c$-coloring of $G$}. Note that we do not restrict the
dimension of the entangled state or the rank of the measurement operators,
we only care about the \textit{number} of measurement operators needed
to win the game with certainty.
\begin{defn}
For all graphs $G$, the \textit{quantum chromatic number} $\chi_{q}(G)$
is the minimum number $c$ such that there exists a quantum $c$-coloring
of $G$.
\end{defn}

Let $\overline{F}$ be the entry-wise complex conjugate of $F$.
It is known \cite{qchrom,SS12} that there always exists a quantum $c$-colouring in a special form:
\begin{defn}
A quantum $c$-coloring of $G$ is in \emph{normal form} if:
\begin{enumerate}
\item All POVMs are projective measurements with $c$ projectors of rank
$r$.
\item The shared state is \mbox{$\ket{\psi} = \frac{1}{\sqrt{rc}} \sum_{i\in [rc]} \ket{i}\ket{i} $}.
\item Alice's projectors are related to Bob's as follows: for all $v,\alpha$,
$E^v_{\alpha}=\overline{F^v_{\alpha}}$.
\item The consistency conditions (\ref{eq:consistency1}) and (\ref{eq:consistency2})
can be expressed as the single condition:
\begin{equation}
\forall(v,w)\in E,\forall\alpha\in[c],\ \Tr(E^v_{\alpha}E^w_{\alpha})=0.\label{eq:consistency_new}
\end{equation}
\end{enumerate}
\end{defn}

\begin{thm}[\cite{qchrom, SS12}]
If $G$ has a quantum $c$-coloring, then
it has a quantum $c$-coloring in normal form.
\label{wlog_strategies}
\end{thm}

\subsection{Relationship with KS sets} \label{sec:qcn2}

We know from \cite{SS12} that \emph{weak} KS sets characterize all the graphs
with a separation between the chromatic number and the rank-1 quantum
chromatic number (\emph{i.e.,} the minimum $c$ such that there exists a quantum
$c$-coloring of the graph using only rank-$1$ projectors). We now
prove that \emph{projective} KS sets characterize all the graphs with a
separation between the chromatic number and the quantum chromatic
number.

Theorem \ref{wlog_strategies} allows us to identify a quantum $c$-coloring with Alice's multiset of projectors, denoted as $\{P^v_{\alpha}\}_{v\in V, \alpha \in [c]}$.
\begin{thm}
\label{thm:chiq_chi} For all graphs $G$, we have that $\chi(G)> \chi_{q}(G) =: c$
if and only if for all quantum $c$-colorings in normal form, $S = \bigcup_{v\in V,\alpha\in[c]}\{P^v_{\alpha}\}$ is a projective KS set. \end{thm}
\begin{proof}
$\Rightarrow )$ Let $\chi(G)> \chi_{q}(G) =: c$ and let $S$
be the union of  Alice's projectors in a quantum $c$-coloring in normal form.
We now show that if $S$ \emph{is not} a projective KS set, then we
can properly $c$-color the graph, contradicting the assumption that
$\chi(G)>c$. If $S$ is not a projective KS set then there exists
a marking function $f:S \rightarrow\{0,1\}$ such that for
all orthogonal $P,P'\in S$ we have $f(P)=0$ or $f(P')=0$. We can
use the function $f$ to $c$-color the graph as follows:
\[
\mbox{color}(v)=\alpha\mbox{ if }f(P^v_{\alpha})=1.
\]
 This is a proper $c$-coloring for the following two reasons. First,
the quantum coloring associates each vertex to a projective measurement,
and since $f$ is a marking function, exactly one projector per measurement is mapped to $1$.
Second, the property of $f$
and the consistency condition (\ref{eq:consistency_new}) ensure that
we never color adjacent vertices with the same color.

$\Leftarrow )$ Let $\chi_{q}(G)=c$ and assume that for all quantum $c$-colorings
 in normal form, the union of Alice's projectors is a projective KS set.
Now suppose, towards a contradiction,
that it is possible to classically $c$-color the graph. Then for each $v\in V$ with classical color $\alpha$,
define the projective measurement \mbox{$\{ P^v_i= \ketbra{i+\alpha}{i+\alpha}\}_{i\in\{0,\dots,c-1\}}$}
(where the addition is modulo $c$). It is easy to see that this is
a valid quantum $c$-coloring, and the union
of its vectors consists of one projective measurement only. Thus it is
not a projective KS set, because you can define a function that maps
$\ketbra{1}{1}$ to $1$ and all other projectors to $0$. This contradicts
the assumption that the union of Alice's projectors is a projective KS set.
\end{proof}

As in Section \ref{sec:zec2}, we remark that Theorem \ref{thm:chiq_chi} also follows from Theorem \ref{thm:ks_pt} in the appendix. However, we again prefer the direct approach to underline the structural relationship between graphs with $\chi(G)>\chi_q(G)$ and orthogonality graphs of projective KS sets.

\section{Relationship between chromatic number and zero-error channel capacity} \label{sec:QCN_ZEC}

Here we use the relationships described in the previous sections to show that every graph with a separation between quantum and classical chromatic number can be used to construct a channel with separation between entanglement-assisted and classical one-shot zero-error capacity. Moreover, using this fact we find a new class of channels with large separation between the one-shot zero-error capacities.

For any two graphs $G$ and $H$ we can define their Cartesian product $G\square H$ as follows. The vertex set of $G\square H$ is  $V(G) \times V(H)$. Two vertices $(v,i)$ and $(w,j)$ are adjacent if either $v=w $ and $(i,j)\in E(H)$ or $(v,w)\in E(G)$ and $i=j$.

The main result of this section needs the following lemmas.

\begin{lem}\label{lem:Vizing}
{\cite{Vizing}} For all graphs $G,H$, the
independence number of their Cartesian product satisfies
\[
\alpha(G\square H)\leq\min\{\alpha(G)\cdot|V(H)|,\alpha(H)\cdot|V(G)|\}.
\]
\end{lem}

\begin{lem}
\label{lem:alpha_lessthan_n}Let $G$ be a graph on $n$ vertices
with $\chi(G)>k$. Then $c_0(G\square K_{k})<n$.\end{lem}
\begin{proof}
The vertex set of $G\square K_{k}$ can be partitioned into $n$ disjoint cliques of
size $k$.
Towards a contradiction, suppose
$c_0(G\square K_{k})\geq n$. Then an independent set of size $n$
must contain exactly one vertex from each clique in the partition.
We can get a $k$-coloring for $G$, as follows: if $(v,i)$
belongs to the independent set, color $v\in E(G)$ with the $i$-th
color. This is a proper coloring because, by definition of the Cartesian
product of graphs, for all $(u,v)\in E(G)$ we have $((u,i),(v,i))\in E(G\square K_{k})$,
and hence $u$ and $v$ will not both get color $i$.
This contradicts the assumption that $\chi(G)>k$.\end{proof}

\begin{lem}
Let $G$ be a graph on $n$ vertices and $\chi_q(G)\leq k$. Then $c^*_0(G\square K_k)=\Theta^*(G\square K_k)=n$.
\label{lem:QuantumCap}
\end{lem}
\begin{proof}
Let $G'=G\square K_k$.
We first show that $c_0^*(G')\geq n$. Note that the vertex set of $G'$ can be partitioned into $n$ disjoint cliques of size $k$.
Now consider an quantum $k$-coloring of $G$ in normal form, $\set{P^v_i}_{v\in V, i\in [k]}$. Assign each projector $P^v_i$ to the vertex $(v,i)$ of $G'$. By the properties of a quantum coloring in normal form, this assignment satisfies the requirements of Lemma \ref{lem:ZEC-graph-partition}. Therefore, $c_0^*(G')\geq n$.

Now note that $\overline{K}_n\boxtimes K_k$ is a subgraph of $G\square K_k$, where $\boxtimes$ denotes the strong product of graphs and $\overline{K}_n$ is the complement of $K_k$. By properties of Lov\'{a}sz $\vartheta$ \cite{Lovasz}, we have that
$$\vartheta(G\square K_k)\leq\vartheta(\overline{K}_n\square K_k)=
\vartheta(\overline{K}_n) \vartheta(K_k)=n.$$
Since Lov\'{a}sz $\vartheta$ is an upper bound for the asymptotic zero error entanglement assisted capacity \cite{Winter,ZEC-theta}, we have $\Theta^*(G\square K_k)\leq n$. Putting everything together gives
$$
  n\leq c_0^*(G\square K_k) \leq \Theta^*(G\square K_k)\leq n
$$
and the lemma follows.
\end{proof}

Combining the results from the above lemmas we obtain the following theorem:
\begin{thm} \label{thm:QCN_ZEC}
Let $G$ be a graph on $n$ vertices with $\chi(G)>\chi_{q}(G)=:k$. Then for $G'=G\square K_{k}$:
\begin{enumerate}
 \item $c_0(G') < c_0^*(G')=n$
 \item $c_0(G')\leq\alpha(G)\cdot k$.
\end{enumerate}

\end{thm}
\begin{proof}
Since $\chi(G)>k$, $\chi_q(G)=k$ and $G$ has $n$ vertices, we have from Lemma \ref{lem:alpha_lessthan_n} and Lemma \ref{lem:QuantumCap} that $c_0(G') < c_0^*(G')=n$, as desired. The second bound follows directly from Lemma \ref{lem:Vizing}.
\end{proof}

Note that the second upper bound of Theorem \ref{thm:QCN_ZEC} is very interesting
in the case $\alpha(G)\cdot\chi_{q}(G)\ll n$. This happens
only when there is a separation between quantum and classical chromatic
number, because for the chromatic number we have $\alpha(G)\cdot\chi(G)\geq n$.
Therefore, as we show in the next section, some graphs with a large separation between quantum and classical chromatic number induce graphs with large separation between entanglement-assisted and classical one-shot zero-error capacity.

\subsection{Hadamard graphs}
In this section we apply the observations made above. Specifically, we isolate a new family of graphs for which the one-shot zero-error entanglement-assisted capacity is exponentially larger than its classical counterpart. Each member of this family is a Cartesian product of a Hadamard graph with a complete graph. Chromatic number and quantum chromatic number are known to be different for Hadamard graphs with a sufficiently large number of vertices \cite{KP05, Avis, Godsil}. Graphs with such a separation are important in the attempt to quantify the information-theoretic gain permitted by the use of shared entanglement. Results in this respect have been reported in \cite{BBG11, LMMOR}. Additionally, despite graphs with a separation seems to be rare, examples and potential characterizations are valuable. These could provide a clearer picture on the complexity of other related parameters in algebraic graph theory, like, for example, vector colorings, rank bounds, etc.

The Hadamard graph $\Omega_{n}$ is  a graph with vertex set $\{\pm1\}^{n}$
and edge set $\{(u,v): \inpc{u}{v}=0\}$.
Hadamard graphs are also known in literature as orthogonality graphs and Deutsch-Jozsa graphs.

\begin{thm}
For all $n>8$ that are divisible by 4, there exists $\eps>0$ such that
\[
\frac{c_0^*(\Omega_{n}\square K_{n})}{c_0(\Omega_{n}\square K_{n})}\geq \frac{1}{n} \left(\frac{2}{2-\epsilon}\right)^n.
%\frac{2^{n}}{(2-\epsilon)^{n}\cdot n}.
\]
\end{thm}
\begin{proof}
It is shown in \cite{Avis} that $\chi_q(\Omega_n)\leq n$ for all $n\in\mathbb{Z}_+$.
Since $|V(\Omega_n)|=2^n$, using Lemma \ref{lem:QuantumCap} we conclude that $c_0^*(\Omega_n\square K_n)\geq 2^n$.

On the other hand from Theorem 1.11 in \cite{FR87}, it follows that for all $n$ divisible by 4, there exists $\epsilon>0$ such that $\alpha(\Omega_n)\leq (2-\epsilon)^n$. Hence, by Lemma~\ref{lem:Vizing}, we have that $c_0(\Omega_n\square K_n)\leq (2-\epsilon)^n \cdot n$.  By putting the two observations together we obtain the desired statement.
\end{proof}

To conclude, we give an example that also for small $n$ we can find a large ratio $\frac{c_0^*(\Omega_{n}\square K_{n})}{c_0(\Omega_{n}\square K_{n})} $. The following properties are proven in \cite{KP05,Avis}:
\begin{enumerate}
\item $\alpha(\Omega_{16})=2304$
\item $\chi(\Omega_{16})\geq29$
\item $\chi_{q}(\Omega_{16})=16$.
\end{enumerate}
Take a channel $\N$ with confusability graph $\Omega_{16}\square K_{16}$.
It follows from Theorem \ref{thm:QCN_ZEC} that $c_0^*(\N)=2^{16}$
while $c_0(\Omega_{16}\square K_{16})\leq\alpha(\Omega_{16})\cdot16=36864$.

\section{Conclusions and open questions}

We formally generalized the concept of Kochen-Specker sets and showed their applications in various nonlocality settings. In particular, we showed that projective KS sets lead to classical channels for which a single use aided by entanglement can transfer more information than a single use without entanglement assistance. We have also shown that projective KS sets completely characterize the graphs for which the quantum chromatic number is strictly smaller than the chromatic number.

Furthermore, we used projective KS sets to relate quantum chromatic number to entanglement-assisted one-shot zero-error capacity. For all channels obtained with our construction the Lov\'{a}sz  theta function is equal to the entanglement-assisted zero-error capacity. Hence, although our construction contributes to shed light on the link between the Lov\'{a}sz theta function and entanglement-assisted capacities, it cannot directly be used to resolve whether or not the entanglement-assisted zero-error capacity equals the Lov\'{a}sz theta function \cite{Winter}.

As an example application of the above discussed construction we exhibited a new class of graphs with an exponential separation between entanglement-assisted and classical one-shot zero-error capacities. In the appendix we showed that projective KS sets are in one-to-one correspondence with a class of pseudo-telepathy games that quantum players can win using projective measurements on maximally entangled state.

A number of questions remain open.
\begin{enumerate}
\item Can generalized KS sets that are not projective KS sets be used to construct pseudo-telepathy games? Furthermore, can a wider class of pseudo-telepathy games be characterized using generalized KS sets than projective ones?
\item Can we use a graph $G$ with $c_0(G)<c^*_0(G)$ to construct a $G'$ with $\chi_q(G')<\chi(G')$? Such a construction would be complementary to Theorem~\ref{thm:QCN_ZEC}.
\item Can generalized KS sets that are not projective KS sets also be used to exhibit separations between entanglement-assisted and Shannon capacity of a channel?
\item Determining the computational complexity of $\chi_q(G)$ as a function of the number of vertices in $G$ is now a long standing open question. Can we use the characterization given in Section~\ref{sec:qcn2} to answer it?
\item In Theorem \ref{thm:ks_pt} in the appendix, we show that projective KS sets correspond precisely to a class of pseudo-telepathy games  that quantum players can win using projective measurements on maximally entangled state. It follows from Theorem \ref{wlog_strategies} that graph coloring games on graphs where $\chi(G) > \chi_q(G)$ are a subclass of such games. Can the whole class be interpreted as pseudo-telepathy games based on some graph parameter?
\end{enumerate}

\begin{paragraph} {Acknowledgements} The authors thank Jop Bri\"{e}t, Will Matthews, Fernando de Melo, David Roberson, Antonis Varvitsiotis, and Ronald de Wolf  for helpful discussions. We also thank the referees of AQIS'12 for useful comments that improved the readability of the paper.
\end{paragraph}

\bibliographystyle{alphaurl}
%\bibliography{ZEC}
\newcommand{\etalchar}[1]{$^{#1}$}

\appendix

\begin{section} {Weak KS sets from projective KS sets} \label{apx:KS}
 We show here that starting from any projective KS set one can construct (usually infinitely many) underlying weak KS sets.
\begin{lem}
Let $S=\set{P_1,\dotsc,P_k}\subset \mathcal{Q}_n$ be a projective KS set. Consider a set of unit vectors $S'\in\C^n$ defined as follows:
$  S' := \bigcup_{i\in [k]}\mathcal{B}_i, $
where $\mathcal{B}_i=\set{v^i_1,\dotsc, v^i_{\mathop{rk}(P_i)}}$ is some orthonormal basis of the support of $P_i$. Then $S'$ is a weak KS set.
\end{lem}

\begin{proof}
Towards a contradiction assume that $S'$ is not a weak KS set. Then there exists a marking function $f':S'\rightarrow\set{0,1}$ such that for any two orthogonal vectors $u,v\in S'$ either $f'(u)=0$ or $f'(v)=0$. From $f'$ let us construct another marking function $f$ for the set $S$. Define $f:S\rightarrow\set{0,1}$ by
$$   f(P_i) := \sum_{j=1}^{\mathop{rk}(P_i)} f'(v^i_j). $$
Let $P_{1},\dotsc,P_{t}\in S$ be such that $\sum_{j\in[t]} P_{j} = I$. Then the corresponding sets of vectors $\mathcal{B}_{j}$ are mutually disjoint.
Moreover, $\mathcal{B}:=\mathcal{B}_{1}\cup\dotsc\cup\mathcal{B}_{t}$ is a complete orthonormal basis. Since $f'$ is a marking function, we have
$$ 1 = \sum_{v\in\mathcal{B}}f'(v) = \sum_{j=1}^t f(P_j) $$
which shows that  $f$ indeed is a marking function.
Consider any $P_i,P_j\in S$ such that $\Tr(P_iP_j)=0$. Since for any two orthogonal vectors $u,v\in S'$ either $f'(u)=0$ or $f'(v)=0$, it follows that $f'$ evaluates to 1 for at most one of the vectors in $\mathcal{B}_i\cup \mathcal{B}_j$. So, from definition of $f$ it follows that either $f(P_i)=0$ or $f(P_j)=0$ and we have reached a contradiction to $S$ being a projective KS set.
\end{proof}
\end{section}

\section{Projective KS sets and pseudo-telepathy games} \label{apx:pt}
This appendix generalizes the results of \cite{KS-PT} concerning the relationship between \emph{weak} KS sets and a class of
pseudo-telepathy games. We show that there is a relationship between \emph{projective} KS sets and a (larger) class of
pseudo-telepathy games.

In what follows we give the canonical definition of nonlocal games. A \emph{nonlocal game} is an experimental setup between a referee and two players, Alice and Bob. (It can also be defined with more players, but we do not consider this case here.)
The game is not adversarial, but the players collaborate with each other. They are allowed to arrange a strategy beforehand, but they are not allowed to communicate during the game. The referee sends Alice an input $x\in X$ and sends Bob an input $y\in Y$,
according to a fixed and known probability distribution $\pi$ on $X \times Y$. Alice and Bob answer with $a\in A$ and $b\in B$ respectively, and the referee declares the outcome of the game according
to a verification function $V: A\times B\times X\times Y \rightarrow \{$win $=1$, lose $=0\}$.
So, the nonlocal game is completely specified by the sets $X,Y,A,B$, a distribution $\pi$, and a  verification function $V$.

A \emph{deterministic classical strategy} is a pair of functions $s_A : X\rightarrow A$ and $s_B: Y\rightarrow B$ for Alice and Bob, respectively.
A \emph{quantum strategy} consists of a shared bipartite entangled state $\ket{\psi}$ and POVMs
$\{P^x_a\}_{a\in A},$ for every $x\in X$ for Alice and $\{P^y_b\}_{b\in B},$ for every $y\in Y$ for Bob. On input $x$, Alice uses
the POVM $\{P^x_a\}_{a\in A}$ to measure her part of the entangled state and Bob does similarly on his input $y$. Alice (resp.\ Bob) answers with $a$ (resp.\ $b$)  corresponding to the obtained measurement outcome. Therefore, the probability to output $a,b$ given $x,y$ is $\Pr(a,b|x,y) = \bra{\psi} P^x_a \otimes P^y_b \ket{\psi} $.

Informally, a nonlocal game is called a \emph{pseudo-telepathy game} if  players sharing the entangled state win with certainty,
while classical players have nonzero probability to lose. More formally:
\begin{defn} [Pseudo-telepathy game]
A nonlocal game with input sets $X,Y$, output sets $A,B$, input distribution $\pi$ and verification function $V$ is called a \emph{pseudo-telepathy game} if:
\begin{enumerate}
 \item There exists a quantum strategy such that for any $(a,b,x,y)$ with $\pi (x,y)> 0$ it holds that $\bra{\psi} P^x_a \otimes P^y_b \ket{\psi} \neq 0$ implies $V(a,b,x,y)=1$.
 \item For all deterministic classical strategies $s_A, s_B$, there exists a tuple $(a,b,x,y)$ with $\pi (x,y)> 0$ such that
$V(a,b,x,y)=0$ but $s_A(x) = a$ and $s_B(y) = b$.
\end{enumerate}
\end{defn}

The following theorem relates projective KS sets and a special kind of pseudo-telepathy games. We will only consider KS sets for which each projector is part of some projective measurement from $S$. Note that any KS set contains a KS set that satisfies this property. Therefore, projectors not contained in any measurement from $S$ are inessential.
Let $\overline{P}$ be the entry-wise complex conjugate of $P$.
For short, we say that the set $S=\set{P^x_a}_{(x,a)} \cup \set{\overline{P^y_b}}_{(y,b)}$ together with $\ket{\psi}$ is a quantum strategy for a nonlocal game $\mathcal G$ if $\{{P^x_a}_{(x,a)}\}$, $\{{P^y_b}_{(y,b)}\}$ and $\ket{\psi}$ are a quantum strategy for $\mathcal G$. We call a strategy optimal if players never lose.
Let $\ket{\Psi} = \frac{1}{\sqrt n} \sum_{i\in [n]} \ket{i}\ket{i}$, where $\{\ket{i}\}_{i\in [n]}$ is the standard basis of $\C^n$.

\begin{thm} \label{thm:ks_pt}
Let $S \subset \mathcal{Q}_n$ be a set of projectors such that every $P\in S$ is contained in some measurement $\mathcal{M}\subset S$. Then $S$ is a projective Kochen-Specker set if and only if there exists a pseudo-telepathy game for which $S$ together with $\ket{\Psi}$ is an optimal quantum strategy.
\end{thm}
\begin{proof}
$\Rightarrow )$ Assume $S$ is a projective KS set. Let $\set{S_1,\dotsc, S_k}$ be the set of all projective measurements contained in $S$ (\emph{i.e.\ }each $S_i$ is a set of projectors that sum to identity). Consider a nonlocal game $\mathcal{G}$ where $X,Y = [k]$, $A,B =[\max_i\abs{S_i}]$, $\pi$ is the uniform distribution, and the verification function is defined by
\begin{align*}
  V(a,b,x,y) = 1 & \Leftrightarrow    \bra{\Psi} P^x_a \otimes \overline{P^y_b} \ket{\Psi} \neq 0 \\
  & \Leftrightarrow  \Tr(P^x_a P^y_b) \neq 0.
\end{align*}

By definition, this game has a optimal quantum strategy in which Alice and Bob measure their part of $\ket{\Psi}$ using projective measurements $S_x$ and $\overline{S_y}$ respectively upon receiving inputs $(x,y)$.

Towards a contradiction, suppose there is an optimal classical strategy $(s_A, s_B)$. Note that $s_A=s_B$ by definition of $V$. Let us now construct a marking function $f$ for $S$ from $s_A$:
\begin{align*}
\forall P\in S,\   f(P)=1 \Leftrightarrow\  & \exists a\in A, x\in X \mbox{ such that } \\
& P=P^x_a \mbox{ and } s_A(x)=a.
\end{align*}
 By definition of $V$, we see that Alice has to pick a unique element $P^x_a$ from each projective measurement $S_x\subset S$ in a consistent way. More precisely, whenever $P^x_a\in S_y$ for some $y$, the element picked from $S_y$ also has to be $P^x_a$.
Hence, $f$ indeed is a valid marking function for $S$. As $S$ is a projective KS set, there exist distinct $P^x_a,P^y_b\in S$ such that $\Tr(P^x_a P^y_b)=0$ and $f(P^x_a)=f(P^y_b)=1$. Then  $s_A(x)=a$ and $s_B(y)=s_A(y)=b$ but $V(a,b,x,y)=0$, since $\Tr(P^x_a P^y_b)=0$. Hence, classical players cannot have a optimal strategy and the desired statement follows.

$\Leftarrow )$
Assume $\mathcal{G}$ is a pseudo-telepathy game and $S=\set{P^x_a}_{(x,a)} \cup \set{\overline{P^y_b}}_{(y,b)}$ together with $\ket{\Psi}$ is a winning quantum strategy. Every marking function $f$ for $S$ can be mapped to a classical strategy in the following way:
$$s_A(x) = a \Leftrightarrow  f(P^a_x) = 1 \text{ and }
  s_B(y) = b \Leftrightarrow  f(\overline{P^b_y}) = 1.$$
Since $\mathcal{G}$ is a pseudo-telepathy game, for every $f$ there exists a tuple $(a,b,x,y)$ such that $s_A(x)=a$ and $s_B(y)=b$ (and therefore $f(P^x_a)=f(P^y_b)=1$), but $V(a,b,x,y)=0$. Since quantum players always win we have $\bra{\Psi}P^x_a \otimes \overline{P^y_b}\ket{\Psi}=0$ and this implies $\Tr(P^x_a P^y_a)=0$  by the properties of $\ket{\Psi}$.
Therefore, for any marking function $f$ for $S$ we can find orthogonal projectors $P^x_a,P^y_b\in S$ such that $f(P^x_a)=f(P^y_b)=1.$ Hence, $S$ is projective KS set.
\end{proof}

\end{document}